\documentclass[envcountsect,envcountsame,runningheads,a4paper]{llncs}

\usepackage{graphicx}
\usepackage{cite}
\usepackage{amssymb,amsmath,amsfonts} 
\usepackage{color}
\usepackage{subfigure}

\renewcommand{\paragraph}[1]{\smallskip\noindent\textit{#1}}

\newtheorem{prop}[theorem]{Proposition}
\newtheorem{obs}[theorem]{Observation}
\newtheorem{corol}[theorem]{Corollary}

\newcommand{\eqdef}{:=}
\newcommand{\eps}{\varepsilon}

\def\Z{{\mathbb Z}}

\newcommand\R{{\mathbb R}}

\title{Interference Minimization in Asymmetric Sensor Networks\thanks{%
KB supported in part by the Netherlands Organisation for Scientific 
Research (NWO) under project no.\ 612.001.207.
WM supported in part by DFG Grants MU 3501/1
and MU 3502/2.
}
}

\author{
Yves Brise\inst{1}
\and
Kevin Buchin\inst{2}
\and Dustin Eversmann\inst{3}
\and Michael Hoffmann\inst{1}
\and Wolfgang~Mulzer\inst{3}
}

\authorrunning{Y.~Brise, K.~Buchin, D.~Eversmann, M.~Hoffmann and W.~Mulzer}

\institute{
ETH Z\"urich, Switzerland,
{\tt hoffmann@inf.ethz.ch}
\and
TU Eindhoven, The Netherlands,
{\tt k.a.buchin@tue.de}
\and
FU Berlin, Germany,
{\tt mulzer@inf.fu-berlin.de}
}
\begin{document}
\maketitle

\begin{abstract}
A fundamental problem in wireless sensor networks is to connect
a given set of sensors while minimizing the \emph{receiver
interference}.
This is modeled as follows: each sensor node corresponds
to a point in $\R^d$ and each \emph{transmission range} corresponds
to a ball. The receiver interference of a sensor node is
defined as the number of transmission ranges it lies in.
Our goal is to choose transmission radii that
minimize the maximum interference while maintaining a
strongly connected asymmetric communication graph.

For the two-dimensional case, we show that it is NP-complete
to decide whether one can achieve a receiver interference
of at most $5$. In the one-dimensional case, we prove that there are
optimal solutions with nontrivial structural properties. These
properties can be exploited to obtain an
exact algorithm that runs in quasi-polynomial time.
This generalizes a result by Tan~et al.~to the asymmetric case.
\end{abstract}
\section{Introduction}

Wireless sensor networks constitute a popular paradigm in mobile networks:
several small independent devices are distributed in a certain region,
and each device has limited computational resources. The devices
can communicate through a wireless network. Since battery life is limited,
it is imperative that the overhead for the communication be kept as small
as possible.
One major concern when trying to achieve this goal is to control the
\emph{interference} caused by competing senders. This enables us
to reduce the range of the senders, thus increasing
battery life. At the same time, we need to
ensure that the resulting communication graph remains connected.

There are many different ways to formalize the problem of
interference minimization.  Usually, the devices are modeled as
points in $d$-dimensional space, and
the transmission ranges are modeled as $d$-dimensional balls.
Each point can choose the radius of its transmission range, and
different choices of transmission ranges lead to different
reachability structures.
There are two ways to interpret the resulting
communication graph. In the \emph{symmetric} case,
the communication graph is undirected, and it contains an edge between
two points $p$ and $q$ if and only if both $p$ and $q$ lie in the transmission
range of the other point. For a valid
assignment of transmission ranges, we require that
the communication graph is connected.
In the \emph{asymmetric} case, the communication graph is directed,
and there is an edge from $p$ to $q$ if and only if $p$ lies in the
transmission range of $q$. We require that the communication graph
is strongly connected, or, in a slightly different model,
that there is one point that is reachable from
every other point through a directed path.

In both the symmetric and the asymmetric case, the
(\emph{receiver-centric}) \emph{interference}
of a point is defined as the number of transmission
ranges that it lies in~\cite{RickenbachWaZo09}. The goal is to find a
valid assignment of transmission ranges that makes the maximum interference
as small as possible. We refer to the resulting interference as \emph{minimum interference}.
The minimum interference under the two models for the 
asymmetric case differs by at most one: if there is a point reachable 
from every other, we can increase its transmission range to 
include all other points. As a result, the communication 
graph becomes strongly connected, while the minimum 
interference increases by at most one.

Let $n$ be the number of points.
In the symmetric case,
one can always achieve interference
$O(\sqrt{n})$, and this is sometimes necessary~\cite{HalldorssonTo08,
RickenbachWaZo09}. In the one-dimensional case, there is an
efficient approximation algorithm with
approximation factor $O(n^{1/4})$~\cite{RickenbachWaZo09}.
Furthermore,
Tan~et al.~\cite{TanLoWaHuLa11}
prove the existence of optimal solutions
with interesting structural properties in one dimension.
This can be used to obtain a nontrivial exact algorithm for
this case.
In the asymmetric case, the interference is significantly smaller:
one can always achieve interference $O(\log n)$, which is sometimes
optimal (e.g.,~\cite{Korman12}).

\paragraph{Our results.}
We consider interference minimization in asymmetric wireless
sensor networks in one and two dimensions. We show that for
two dimensions, it is NP-complete to find a valid
assignment that minimizes the maximum interference.
In one dimension we consider our second model requiring
one point that is reachable from
every other point through a directed path.
 Generalizing the
result by Tan~et al.~\cite{TanLoWaHuLa11}, we show that
there is an optimal solution that
exhibits a certain binary tree structure. By means of
dynamic programming, this structure can
be leveraged for a nontrivial exact algorithm. Unlike the
symmetric case, this algorithm always runs in quasi-polynomial
time $2^{O(\log^2 n)}$, making it unlikely that the
one-dimensional problem is NP-hard.
Nonetheless, a polynomial time algorithm remains elusive.

\section{Preliminaries and Notation}

We now formalize our interference model for the
planar case.
Let $P \subset \R^2$ be a planar $n$-point set.
A \emph{receiver assignment} $N : P \rightarrow P$
is a function that assigns to each point in $P$ the furthest
point that receives data from $P$.
The resulting (asymmetric) \emph{communication graph}
$G_P(N)$ is the directed graph
with vertex set $P$ and edge set
$E_P(N) = \{(p, q) \mid \|p - q\| \leq  \|p - N(p)\|\}$,
i.e., from each point $p \in P$ there are edges
to all points that are at least as close as the
assigned receiver $N(p)$.
The receiver assignment $N$ is \emph{valid}
if $G_P(N)$ is strongly connected.

For $p \in \R^2$ and $r > 0$, let $B(p, r)$ denote
the closed disk with center $p$ and radius $r$.
We define $B_P(N) = \{B(p, d(p, N(p)) \mid p \in P \}$ as the
set that contains for each $p \in P$
a disk with center $p$ and $N(p)$ on the boundary.
The disks in $B_P(N)$ are called the \emph{transmission
ranges} for $N$.
The \emph{interference}
of $N$, $I(N)$, is the maximum number of transmission ranges
that cover a point in $P$, i.e.,
$I(N) = \max_{p \in P} |\{ p \in B \mid B \in B_P(N) \}|$.
In the \emph{interference minimization problem}, we are looking for
a valid receiver assignment with minimum interference.

\section{NP-completeness in Two Dimensions}

We show that the following problem is
NP-complete: given a planar point set
$P$, does there exist a valid receiver assignment
$N$ for $P$ with $I(N) \leq 5$?
It follows that the minimum interference for
planar point sets is
NP-hard to approximate within a
factor of $6/5$.

The problem is clearly in NP.
To show that interference minimization is NP-hard,
we reduce from the problem of deciding
whether a grid graph of maximum degree $3$ contains
a Hamiltonian path:
a \emph{grid graph} $G$ is a graph whose vertex
set $V \subset \Z \times \Z$ is a finite
subset of the integer grid.
Two vertices $u, v \in V$ are adjacent in $G$
if and only if $\|u-v\|_1=1$, i.e., if $u$ and $v$
are neighbors in the integer grid.
A \emph{Hamiltonian path} in $G$ is a path
that visits every vertex in $V$ exactly once.
Papadimitriou and Vazirani showed that
it is NP-complete to decide whether a grid graph $G$
of maximum degree $3$ contains a Hamiltonian
cycle~\cite{PapadimitriouVa84}.
Note that we may assume that $G$
is connected;
otherwise there can be no Hamiltonian path.

Our reduction proceeds by replacing each vertex $v$
of the given grid graph $G$ by a \emph{vertex gadget} $P_v$;
see Fig.~\ref{fig:gadget_bare}.
\begin{figure}[b]
  \centering
  \includegraphics[scale=0.8]{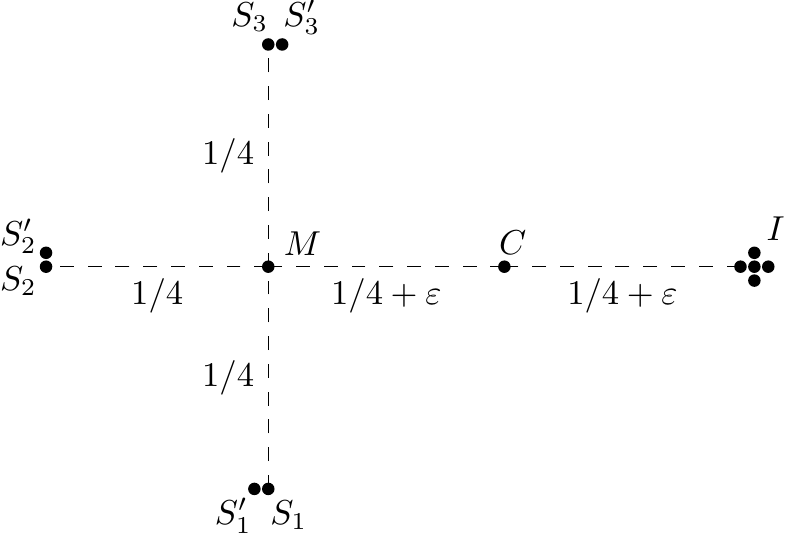}
  \caption{The vertex gadget. }
  \label{fig:gadget_bare}
\end{figure}
The vertex gadget consists of 13 points, and it has five parts:
(a) the \emph{main point} $M$ with
the same coordinates as $v$;
(b) three \emph{satellite stations}
with two points each: $S_1$, $S_1'$,
$S_2$, $S_2'$, $S_3$, $S_3'$.
The coordinates of the $S_i$ are
chosen from $\{v\pm (0,1/4), v \pm (1/4,0)\}$
so that there is a satellite station for each
edge in $G$ that is incident to $v$. If $v$ has degree two, the third 
satellite station can be placed in any of the two remaining directions.
The $S'_i$ lie
at the corresponding clockwise positions from
$\{v\pm (\eps,1/4), v \pm (1/4,-\eps)\}$,
for a sufficiently small $\eps > 0$;
(c) the \emph{connector} $C$, a point
that lies roughly at the remaining position from
$\{v \pm (0,1/4), v \pm (1/4,0)\}$ that is
not occupied by a satellite station, but
an $\eps$-unit further away from $M$. For example,
if $v+(0,1/4)$ has no satellite station, then $C$ lies
at $v+(0,1/4+\eps)$; and
(d) the \emph{inhibitor}, consisting of five
points $I_c, I_1, \dots, I_4$.
The point $I_c$ is the center of the inhibitor
and $I_1$ is the point closest to $C$.
The position of $I_c$ is
$M + 2(C-M) +\eps(C-M)/\|C-M\|$, that is,
the distance between $I_c$ and
$C$ is an $\eps$-unit larger than the distance  between $C$
and $M$:
$\|M-C\|+\eps = \|C-I_c\|$.
The points $I_1, \dots, I_4$ are placed at
the positions $\{I_c \pm (0,\eps), I_c \pm (\eps,0)\}$,
with $I_1$ closest to $C$.

\begin{figure*}
  \centering
  \includegraphics[scale=0.6]{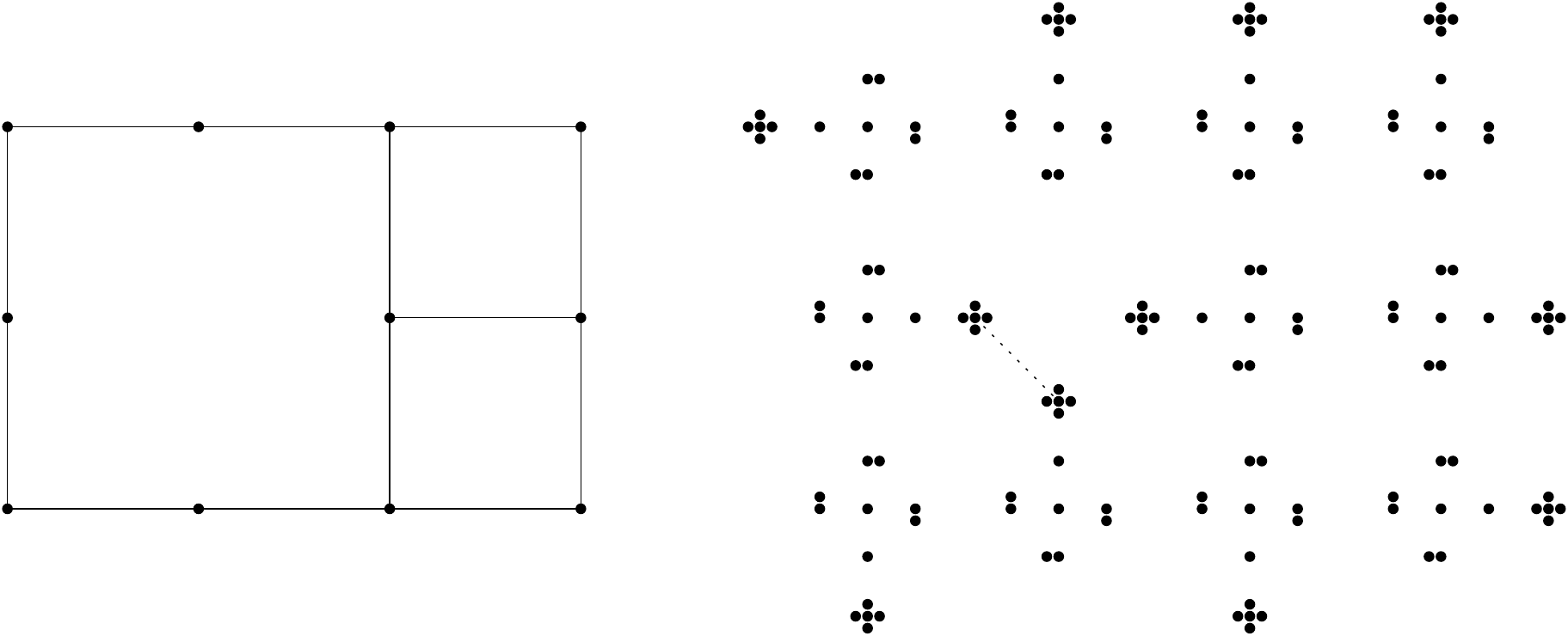}
  \caption{An example reduction.}
  \label{fig:gadget_example}
\end{figure*}

Given a grid graph $G$, the reduction can be
carried out in polynomial time: just replace each
vertex $v$ of $G$ by the corresponding gadget $P_v$; see
Fig.~\ref{fig:gadget_example} for an example.
Let $P = \bigcup_{v \in G} P_v$ be the resulting point set.
Two satellite stations in $P$ that correspond to the same
edge of $G$ are called \emph{partners}.
First, we investigate the interference in any valid
receiver assignment for $P$.

\begin{lemma}
\label{lem:gadget_NN}
Let $N$ be a valid receiver assignment for $P$.
Then in each vertex gadget, the points $I_c$ and $M$
have interference as least $5$,
and the points $S_1, S_2$, and $S_3$ have interference at
least $3$.
\end{lemma}

\begin{proof}
For each point $p \in P$, the transmission
range $B(p, d(p, N(p))$ must contain
at least the nearest
neighbor of $p$. Furthermore, in each
satellite station and in each inhibitor,
at least one point must have an assigned receiver outside
of the satellite station or inhibitor; otherwise,
the communication graph $G_P(N)$ would not be strongly connected.
This forces interference of $5$ at $M$ and at $I_c$: each satellite
station and $C$ must have
an edge to $M$, and $I_1, \dots, I_4$ all must have an
edge to $I_c$. Similarly, for $i = 1, \dots 3$, the main
point $M$ and the satellite $S_i'$ must have an edge to $S_i$;
see Fig.~\ref{fig:gadget_NN+}.
\qed{}
\end{proof}

\begin{figure}[htbp]
  \centering
  \includegraphics[scale=0.8]{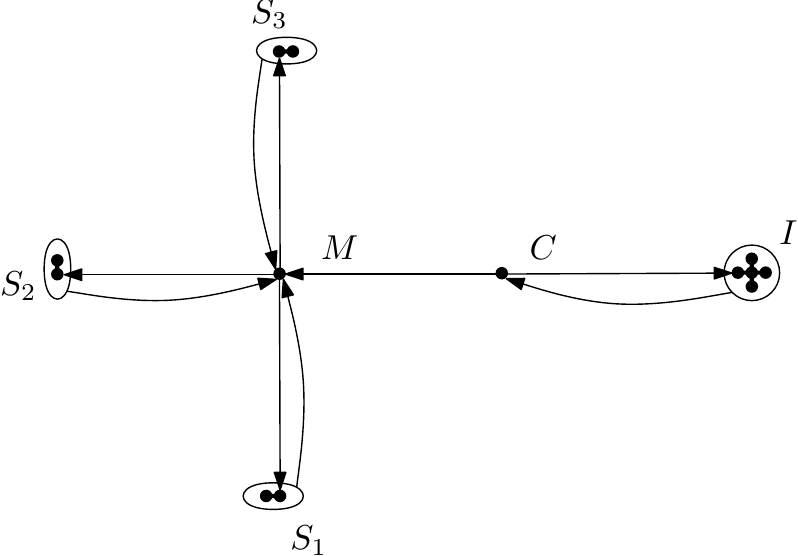}
  \caption{The nearest
    neighbors in a vertex gadget. }
  \label{fig:gadget_NN+}
\end{figure}

Let $N$ be a valid receiver assignment,  and
let $P_v$ be a vertex gadget in $P$. An \emph{outgoing}
edge for $P_v$ is an edge in $G_P(N)$ that originates
in $P_v$ and ends in a different vertex gadget.
An \emph{incoming} edge for $P_v$ is an edge that
originates in a different gadget and ends in $P_v$.
A \emph{connecting} edge for $P_v$ is either an outgoing
or an incoming edge for $P_v$. If $I(N) \leq 5$ holds, then
Lemma~\ref{lem:gadget_NN} implies that a connecting
edge can be incident only to satellite
stations. The proof of the following lemma is given in Appendix~\ref{sec:omitted}.

\begin{lemma}\label{lem:outedge}
Let $N$ be a valid receiver assignment for $P$
with $I(N) \leq 5$.
Let $P_v$ be a vertex gadget of $P$ and
$e$ an outgoing edge from $P_v$
to another vertex gadget $P_w$.
Then $e$ goes from a satellite station of $P_v$ to
its partner satellite station in $P_w$.
Furthermore, in each satellite station of $P_v$,
at most one point is incident to outgoing edges.
\end{lemma}

\noindent
Next, we show that the edges between the vertex gadgets are
quite restricted.

\begin{lemma}
\label{lem:inout}
Let $N$ be a valid receiver assignment for
$P$ with $I(N) \leq 5$.
For every vertex gadget $P_v$ in $P$,
at most two satellite stations in $P_v$
are incident to connecting edges in $G_P(N)$.
\end{lemma}

\begin{proof}
By Lemma~\ref{lem:outedge} connecting edges are between satellite stations and
by Lemma~\ref{lem:gadget_NN},
the satellite points $S_i$ in $P_v$ have
interference at least $3$.

First, assume that all three satellite
stations in $P_v$ have outgoing edges.
This would increase the interference
at all three $S_i$ to $5$. Then, $P_v$ could
not have any incoming edge from another vertex gadget,
because this would increase the interference for
at least one $S_i$ (note that due to the placement of
the $S_i'$, every incoming edge causes interference at
an $S_i$). If $P_v$ had no incoming edge, $G_P(N)$ would
not be strongly connected.
It follows that  $P_v$ has at most two satellite stations
with outgoing edges.

Next, assume that two satellite stations in
$P_v$ have outgoing edges. Then, the third satellite station
of $P_v$ cannot have an incoming edge,
as the two outgoing edges already increase the interference at the
third satellite station to $5$.

Hence, we know that every vertex gadget $P_v$ either
(i) has connecting edges with all three partner
gadgets, exactly one of which is outgoing,
or  (ii) is connected to
at most two other vertex gadgets. Take a vertex gadget $P_v$ of type (i)
with partners $P_{u_1}, P_{u_2}, P_w$.
Suppose that $P_v$ has incoming edges from $P_{u_1}$ and
$P_{u_2}$
and that the outgoing edge goes to $P_w$. Follow the outgoing edge to
$P_w$. If $P_w$ is of type (i), follow the outgoing
edge from $P_w$; if $P_w$ is of type (ii) and has an
outgoing edge to a vertex gadget we have not seen yet, follow this edge.
Continue this process until $P_v$ is reached again or until
the next vertex gadget has been visited already. This gives
all vertex gadgets that are reachable from $P_v$ on a
directed path.
However, in each step there is only one choice for the next
vertex gadget. Thus,
the process cannot discover $P_{u_1}$ and $P_{u_2}$, since both of them
would lead to $P_v$ in the next step, causing the process to stop.
It follows that at least one of $P_{u_1}$ or
$P_{u_2}$ is not reachable from $P_v$, although
$G_P(N)$ should be strongly connected. Therefore, all vertex gadgets in
$G_P(N)$ must be of type (ii), as claimed in the lemma.
\qed{}
\end{proof}

\noindent
We can now prove the main theorem of this section.

\begin{theorem}\label{thm:asym_NPC}
Given a point set $P \subset \R^2$, it is
\textup{NP}-complete to decide whether
there exists a valid receiver assignment $N$ for
$P$ with $I(N) \leq 5$.
\end{theorem}

\begin{proof}
Using the receiver assignment $N$ as certificate,
the problem is easily seen to be in NP.
To show NP-hardness, we use the polynomial time reduction
from the Hamiltonian path problem in grid graphs:
given a grid graph $G$ of maximum degree
$3$, we construct a planar point set $P$ as above.
It remains to verify that $G$ has a Hamiltonian path
if and only if $P$ has a valid receiver assignment $N$
with $I(N) \leq 5$.

Given a Hamilton path $H$ in
$G$, we construct
a valid receiver assignment $N$ for
$P$ as follows:
in each vertex gadget, we set $N(M) = C$,
$N(C) = M$, and $N(I_1) =C $. For
$i = 1, \dots, 3$ we set $N(S'_i) = S_i$ and
$N(I_{i+1}) = I_c$.
Finally, we set $N(I_c) = I_1$.
This essentially creates the edges from
Fig.~\ref{fig:gadget_NN+}, plus the edge
from $M$ to $C$.
Next, we encode $H$ into $N$: for each $S_i$ on
an edge of $H$, we set $N(S_i)$  to the corresponding
$S_i$ in the partner station. For the remaining
$S_i$, we set $N(S_i) = M$.
Since $H$ is Hamiltonian,
$G_P(N)$ is
strongly connected (note that
each vertex gadget induces a strongly connected
subgraph).
It can now be verified that $M$ and $I_c$ have
interference $5$;
$I_2$, $I_3$, $I_4$ have interference $2$; and
$I_1$ has interference $3$.
The point $C$ has interference between $2$ and
$4$, depending on whether $S_1$ and $S_3$ are
on edges of $H$.
The satellites $S_i$ and $S'_i$
have interference at most $5$ and $4$,
respectively.

Now consider a valid receiver assignment $N$ for
$P$ with
$I(N) \leq 5$.
Let $F$ be the set of edges in $G$ that correspond
to  pairs of vertex gadgets with a connecting edge
in $G_P(N)$.
Let $H$ be the subgraph that $F$ induces in $G$.
By Lemma~\ref{lem:inout}, $H$ has maximum degree $2$.
Furthermore, since $G_P(N)$ is strongly connected,
the graph $H$ is connected and meets all vertices of $G$.
Thus, $H$ is a Hamiltonian path (or cycle) for $G$,
as desired.
\qed{}
\end{proof}

\textbf{Remark}.  A similar result to Theorem~\ref{thm:asym_NPC}
also holds for symmetric communication graphs networks~\cite{Buchin08}.

\section{The One-Dimensional Case}

For the one-dimensional case we
minimize receiver interference under the second model discussed in the introduction:
given $P \subset \R$ and
a receiver assignment
$N: P \rightarrow P$, the graph
$G_P(N)$ now has a directed edge from each point $p \in P$
to its assigned receiver $N(p)$, and no other edges.
$N$ is \emph{valid} if $G_P(N)$ is acyclic and if
there is a sink $r \in P$ that is reachable from every
point in $P$. The sink has no outgoing edge. The interference
of $N$, $I(N)$,  is defined as before.

\subsection{Properties of Optimal Solutions}

We now explore the structure of optimal
receiver assignments. Let $P \subset \R$ and $N$ be a valid receiver
assignment for $P$ with sink $r$. We can interpret
$G_P(N)$ as a directed tree, so we
call $r$ the \emph{root} of $G_P(N)$.
For a directed edge $pq$ in $G_P(N)$, we say that $p$ is a \emph{child} of $q$
and $q$ is the \emph{parent} of $p$.
We write $p \leadsto_N q$ if there is a
directed path from $p$ to $q$ in $G_P(N)$.
If $p \leadsto_N q$,
then $q$ is an \emph{ancestor} of $p$ and
$p$ a
\emph{descendant} of $q$. Note that $p$ is both an
ancestor and a descendant of $p$. Two points $p, q \in P$ are \emph{unrelated}
if $p$ is neither an ancestor nor a descendant of $q$.
For two points $p$, $q$, we define $((p,q)) = (\min\{p,q\}, \max\{p,q\})$
as the open interval bounded by $p$ and $q$, and
$[[p,q]] = [\min\{p,q\}, \max\{p,q\}]$ as the closure of
$((p, q))$.
An edge $pq$ of $G_P(N)$ is a \emph{cross edge}
if the interval $((p,q))$ contains at least one point that is
not a descendant of $p$.

Our main structural result is that there is always an optimal
receiver assignment for $P$ without cross edges.
A similar property was observed by Tan~et al.~for the
symmetric case~\cite{TanLoWaHuLa11}.

\begin{lemma}\label{lem:nocross}
Let $N^*$ be a valid receiver assignment for $P$ with
minimum interference.
There is a valid assignment $\widetilde{N}$ for $P$ with
$I(\widetilde{N}) = I(N^*)$
such that $G_P(\widetilde{N})$ has no cross edges.
\end{lemma}

\begin{proof}
Pick a valid assignment
$\widetilde{N}$ with minimum interference that minimizes the total
length of the cross edges
\[
C(\widetilde{N}) \eqdef \sum_{pq \in C} \|p-q\|,
\]
where $C$ are the cross-edges of $G_P(\widetilde{N})$.
If $C(\widetilde{N}) = 0$, we are done.
Thus, suppose $C(\widetilde{N}) > 0$. Pick a cross edge $pq$ such that
the hop-distance (i.e., the number of edges) from $p$ to the root
is maximum among all cross edges.
Let $p_l$ be the leftmost and $p_r$ the rightmost descendant of $p$. We refer 
to Appendix~\ref{sec:omitted} for a proof of the following lemma.

\begin{prop}\label{prop:descendants}
The interval $[p_l, p_r]$ contains only descendants of $p$.
\end{prop}

Let $R$ be the points in $((p,q))$ that
are not descendants of $p$. Each point in $R$
is either unrelated to $p$, or it is an ancestor of $p$.
Let $z \in R$ be the point in $R$ that is closest
to $p$ (i.e.,  $z$ either lies directly to the left of $p_l$ or directly
to the right of $p_r)$.
We now describe how to construct a new valid assignment $\widehat{N}$,
from which we will eventually derive a contradiction to the choice of
$\widetilde{N}$.
The construction is as follows:
replace the edge $pq$ by $pz$. Furthermore, if
(i) $q \leadsto_{\widetilde{N}} z$;
(ii) the last edge  $z'z$ on the path
from $q$ to $z$ crosses the interval $[p_l, p_r]$; and (iii)
$z'z$ is not a cross-edge, we also
change the edge $z' z$ to the edge that connects $z'$ to the
closer of $p_l$ or $p_r$. We give the proof of the following proposition in Appendix~\ref{sec:omitted}.

\begin{prop}\label{clm:validass}
$\widehat{N}$ is a valid assignment.
\end{prop}

\begin{prop}\label{clm:optass}
We have $I(N^*) = I(\widehat{N})$.
\end{prop}

\begin{proof}
Since the new edges are shorter than the edges they replace, each
transmission range for $\widehat{N}$ is contained in the corresponding
transmission range for $\widetilde{N}$. The interference cannot decrease since $N^*$ is
optimal.
\qed{}
\end{proof}

\begin{prop}\label{clm:lesscross}
We have $C(\widehat{N}) < C(\widetilde{N})$.
\end{prop}

\begin{proof}
First, we claim that $\widehat{N}$ contains
no new cross edges, except possibly $pz$:
suppose $ab$ is a cross edge
of $G_P(\widehat{N})$, but not of $G_P(\widetilde{N})$.
This means that $((a,b))$ contains a point $x$
with $x \leadsto_{\widetilde{N}} a$,
but $x \not\leadsto_{\widehat{N}} a$.
Then $x$
must be a descendant of $p$ in $G_P(\widetilde{N})$ and in
$G_P(\widehat{N})$, because
as we saw in the proof of Claim~\ref{clm:validass},
for any $y \in P \setminus [p_l, p_r]$, we have that if $y \leadsto_{\widetilde{N}} a$,
then $y \leadsto_{\widehat{N}} a$.

Hence, $((a,b))$ and $[p_l, p_r]$ intersect.
Since $ab$ is a cross edge, the choice of $pq$ now implies that
$[p_l, p_r] \subseteq ((a,b))$.
Thus, $z$ lies in $[[a,b]]$, because $z$ is a direct neighbor of $p_l$ or $p_r$.
We claim that $b = z$. Indeed, otherwise we would have $z \leadsto_{\widetilde{N}} a$
(since $ab$ is not a cross edge in $G_P(\widetilde{N})$), and thus also
$z \leadsto_{\widehat{N}} a$. However, we already observed
$x \leadsto_{\widehat{N}} p$, so we would have
$x \leadsto_{\widehat{N}} a$ (recall that we introduce the edge $pz$
in $\widehat{N}$).
This contradicts our choice of $x$.

Now it follows that $ab = az$ is the last
edge on the path from $p$ to $z$, because if $a$ were not an ancestor
of $p$, then $ab$ would already be a cross-edge in $G_P(\widetilde{N})$.
Hence, (i) $a$ is an ancestor of $q$; (ii) $az$ crosses the interval
$[p_l, p_r]$; and (iii) $az$ is not a cross edge in $\widetilde{N}$.
These are the conditions for the edge $z' z$ that
we remove from $\widetilde{N}$.
The new edge $e$ from $a$ to $p_l$ or $p_r$ cannot be a cross edge,
because $ab$ is not a cross edge in $G_P(\widehat{N})$ and $e$ does
not cover any descendants of $p$.

Hence, $G_P(\widehat{N})$  contain no new cross-edges,
except possibly $pz$ which replaces
the cross edge $pq$.
By construction, $\|p-z\| < \|p-q\|$, so
$C(\widehat{N}) < C(\widetilde{N})$.
\qed{}
\end{proof}

\noindent
Propositions~\ref{clm:validass}--\ref{clm:lesscross}
yield a contradiction to the choice of $\widetilde{N}$.
It follows that we must have $C(\widetilde{N}) = 0$, as desired.
\qed{}
\end{proof}

Let $P \subset\R$.
We say that a valid assignment $N$ for $P$ has the \emph{BST-property} if the
the following holds for any vertex $p$ of $G_P(N)$:
(i) $p$ has at most one child $q$ with $p < q$ and at most
one child $q$ with $p > q$; and (ii) let $p_l$ be the
leftmost and $p_r$ the rightmost descendant of $p$.
Then $[p_l, p_r]$ contains only descendants of $p$.
In other words: $G_P(N)$ constitutes a binary
search tree for the (coordinates of the) points in $P$.
A valid assignment without cross edges has
the BST-property. The following is therefore an immediate consequence of Lemma~\ref{lem:nocross}.

\begin{theorem}\label{thm:BST}
Every $P \subset \R$ has an optimal valid assignment
with the \emph{BST-property}.\qed
\end{theorem}

\subsection{A Quasi-Polynomial Algorithm}\label{sec:qalgo}

We now show how to use Theorem~\ref{thm:BST} for
a quasi-polynomial time algorithm to minimize the
interference. The algorithm uses dynamic
programming. A subproblem $\pi$ for the dynamic program consists of four
parts: (i) an interval $P_\pi \subseteq P$ of \emph{consecutive} points
in $P$;
(ii) a root $r_\pi \in P_\pi$; (iii) a set
$I_\pi$ of \emph{incoming interference}; and (iv) a set $O_\pi$
of \emph{outgoing interference}.

The objective of $\pi$ is to find an optimal valid assignment $N$ for
$P_\pi$ subject to
(i) the root of $G_N(P_\pi)$ is $r$;
(ii) the set $O_\pi$ contains all transmission ranges of
$B_{P_\pi}(N)$ that cover points in $P \setminus P_\pi$ plus potentially a
transmission range with center $r_\pi$;
(iii) the set $I_\pi$ contains
transmission ranges that cover points in $P_\pi$ and have their center in $P \setminus P_\pi$.
The interference of $N$ is defined as the maximum number of transmission ranges in
$B_{P_\pi}(N) \cup I_\pi \cup O_\pi$ that cover any given point of $P_\pi$.
The transmission ranges in $O_\pi \cup I_\pi$ are given as pairs $(p,q) \in P^2$, where
$p$ is the center and $q$ a point on the boundary of the range.

Each range in $O_\pi \cup I_\pi$ covers a boundary point of
$P_\pi$.
Since it is known that there is always an assignment with interference
$O(\log n)$ (see \cite{RickenbachWaZo09} and Observation~\ref{obs:nna_upper}),
no point of $P$ lies in more than $O(\log n)$
ranges of $B_P(N^*)$. Thus, we can assume that
$|I_\pi \cup O_\pi| = O(\log n)$, and the total number
of subproblems is $n^{O(\log n)}$.

A subproblem $\pi$ can be solved recursively as follows.
Let $A$ be the points in $P_\pi$ to the left of $r_\pi$,
and $B$ the points in $P_\pi$ to the right of $r_\pi$.
We enumerate all pairs $(\sigma, \rho)$ of subproblems with
$P_{\sigma} = A$ and $P_{\rho} = B$,
and we connect the roots $r_\sigma$ and $r_\rho$ to $r_\pi$.
Then we check whether $I_\pi$, $O_\pi$, $I_\sigma$, $O_\sigma$,
$I_\rho$, and $O_\rho$ are \emph{consistent}.
This means that  $O_\sigma$ contains all ranges from $O_\pi$ with center in
$A$ plus the range for the edge $r_\sigma r_\pi$
(if it does not lie in $O_\pi$ yet).
Furthermore,
$O_\sigma$ may contain additional ranges with center in $A$ that cover
points in $P_\pi \setminus A$ but not in $P \setminus P_\pi$.
The set $I_\sigma$ must contain all ranges in $I_\pi$ and $O_\rho$
that cover points in $A$, as
well as the range from $O_\pi$ with center $r_\pi$, if it exists and if it
covers a point in $A$.
The conditions for $\rho$ are analogous.

Let $N_\pi$ be the valid assignment for $\pi$ obtained
by taking optimal valid assignments $N_\sigma$ and $N_\rho$
for $\sigma$ and $\rho$ and
by adding edges from $r_\sigma$ and $r_\rho$ to $r_\pi$.
The interference of $N_\pi$ is then defined with respect to
the ranges in $B_{P_\pi}(N_\pi) \cup I_\pi$ plus the range with
center $r_\pi$ in $O_\pi$ (the other ranges of $O_\pi$ must lie in
$B_{P_{\pi}}(N_\pi)$. We take the
pair $(\sigma, \rho)$ of subproblems which minimizes this interference.
This step takes $n^{O(\log n)}$ time, because the number
of subproblem pairs is $n^{O(\log n)}$ and the overhead per pair is
polynomial in $n$.

The recursion ends if $P_\pi$ contains a single point $r_\pi$. If $O_\pi$
contains only one range, namely the edge from $r_\pi$ to its parent, the
interference of $\pi$ is given by $|I_\pi|+1$.
If $O_\pi$ is empty or contains more than one range, then
the interference for $\pi$ is $\infty$.

To find the overall optimum, we start the recursion with
$P_\pi = P$, $O_\pi=I_\pi=\emptyset$ and every possible root, taking the
minimum of all results. By implementing the recursion with dynamic programming,
we obtain the following result.

\begin{theorem}\label{thm:qp_algo}
Let $P \subset \R$ with $|P| = n$. The optimum interference of $P$ can be
found in time $n^{O(\log n)}$.\qed
\end{theorem}

Theorem~\ref{thm:qp_algo} can be improved slightly. The number of subproblems
depends on the maximum number of transmission ranges that cover the boundary points
of $P_\pi$ in an optimum assignment. This number is bounded by the optimum
interference of $P$.  Using exponential search, we get the following
theorem.
\begin{theorem}
Let $P \subset \R$ with $|P|= n$. The optimum
interference $\textup{OPT}$ for $P$ can be found in time
$n^{O(\textup{OPT})}$.\qed
\end{theorem}

\section{Further Structural Properties in One Dimension}\label{sec:further}

In this section, we explore further structural properties of
optimal valid receiver assignments for one-dimensional point sets.
It is well known that for any $n$-point set $P$, there always exists a valid assignment $\widetilde{N}$ with
$I(\widetilde{N}) = O(\log n)$. Furthermore, there exist point sets such that
any valid assignment $N$ for them must have $I(N) = \Omega(\log n)$~\cite{RickenbachWaZo09}.
For completeness, we include proofs for these facts in Section~\ref{sec:nna}. 
Below we show that there may be an arbitrary number
of left-right turns in an optimal solution. To the best of our knowledge, this result is new, and
it shows that in a certain sense, Theorem~\ref{thm:BST} cannot be improved.

In Theorem~\ref{thm:BST} we proved that there always exists an
optimal solution with the BST-property. Now, we will
show that the structure of an optimal solution cannot
be much simpler than that. Let $P \subset \R$ be finite
and let $N$ be a valid receiver assignment for $P$. A \emph{bend} in $G_P(N)$
is an edge between two non-adjacent points. We will show that
for any $k$ there is a point set $Q_k$ such that any optimal
assignment for $Q_k$ has at least $k$ bends.

For this, we inductively define sets $P_0$, $P_1$, $\ldots$ as follows.
For each $P_i$, let $\ell_i$ denote the diameter of $P_i$.
$P_0$ is just the origin (and $\ell_0 = 0$). Given $P_i$, we let $P_{i+1}$
consist of two copies of $P_i$, where the second copy is translated
by $2\ell_i + 1$ to the right, see Fig.~\ref{fig:lowerb}. By induction,
it follows that
$|P_i| = 2^i$ and $\ell_i = (3^i-1)/2$.

\begin{figure}
\centering
\includegraphics[scale=0.8]{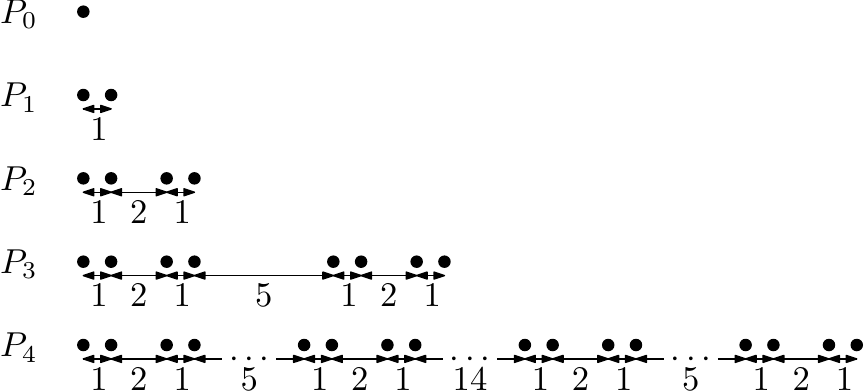}
\caption{Inductive construction of $P_i$.}
\label{fig:lowerb}
\end{figure}

\begin{prop}\label{prop:lowerb}
Every valid assignment for $P_i$ has interference at least
$i$.
\end{prop}

We give the proof of Proposition~\ref{prop:lowerb} in Appendix~\ref{sec:omitted}.

\begin{lemma}\label{lem:opt_ass}
For $i \geq 1$, there exists a valid assignment $N_i$ for
$P_i$ that achieves interference $i$. Furthermore, $N_i$ can
be chosen with the following properties: (i) $N_i$ has the \textup{BST}-property;
(ii) the leftmost or the rightmost point of $P_i$ is the root
of $G_{P_i}(N_i)$; (iii) the interference at the root is
$1$, the interference at the other extreme point of $P_i$ is
$i$.
\end{lemma}

\begin{proof}
We construct $N_i$ inductively. The point set $P_1$
has two points at distance $1$, so any valid
assignment has the claimed properties.

Given $N_i$, we construct $N_{i+1}$: recall that $P_{i+1}$
consists of two copies of $P_i$ at distance
$\ell_i + 1$. Let $L$ be the left and $R$ the right copy.
To get an assignment $N_{i+1}$ with the leftmost point
as root, we use the assignment $N_i$ with the left point as root
for $L$ and for $R$,
and we connect the root of $R$ to the rightmost
point of $L$. This yields a valid assignment.
Since the distance between $L$ and $R$
is $\ell_i + 1$, the interference for all points in
$R$ increases by $1$.  The interferences for
$L$ do not change, except for the rightmost point,
whose interference increases by $1$. Since $|L| \geq 2$,
the desired properties follow by induction.
The assignment with the rightmost point as root is constructed
symmetrically.
\qed{}
\end{proof}

The point set $Q_k$ is constructed recursively.
$Q_0$ consists of a single point $a = 0$ and a copy $R_2$
of $P_2$ translated to the right by $\ell_2 + 1$ units.
Let $d_{k-1}$ be the diameter of $Q_{k-1}$.
To construct $Q_k$ from $Q_{k-1}$, we add a copy $R_{k+2}$ of
$P_{k+2}$, at distance $d_{k-1} + 1$ from $Q_k$. If $k$ is odd,
we add $R_{k+2}$ to the left, and if $k$ is even,
we add $R_{k+2}$ to the right; see Fig.~\ref{fig:bends}.
\begin{figure}
\centering
\includegraphics[scale=0.7]{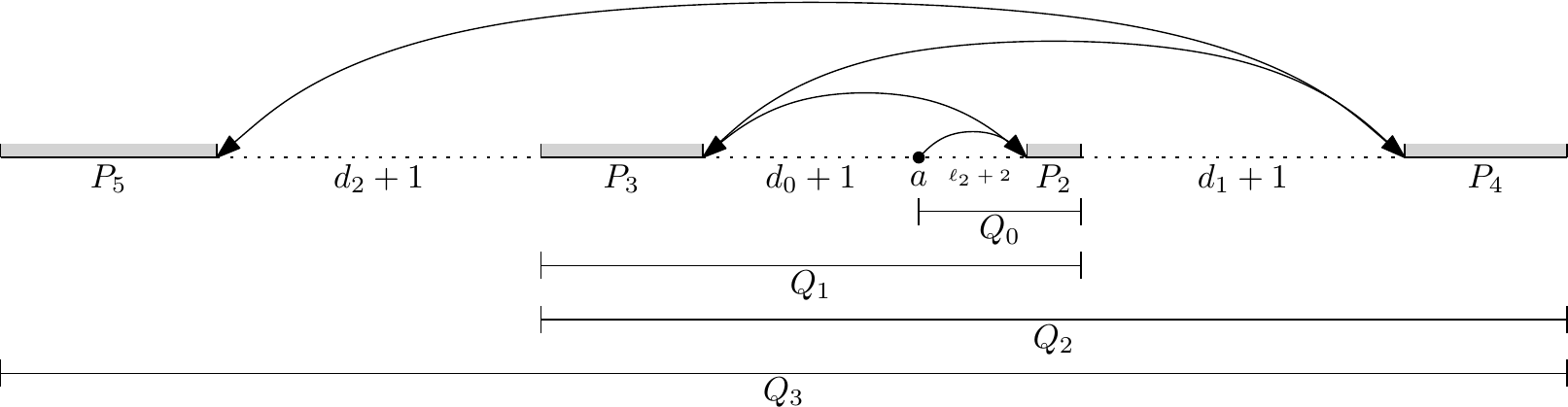}
\caption{The structure of $Q_3$. The arrows indicate the
bends of an optimal assignment.}
\label{fig:bends}
\end{figure}

\begin{theorem}
We have the following properties: (i) the diameter
$d_k$ is $(3^{k+3}-2^{k+3}-1)/2$;
(ii) the optimum interference of $Q_k$ is $k + 2$; and (iii) every optimal
assignment for $Q_k$ has at least $k$ bends.
\end{theorem}

\begin{proof}
By construction, we have $d_0 = 9$ and  $d_k = 2d_{k-1} + 1 + \ell_{k+2}$,
for $k \geq 1$.  Solving the recursion yields the claimed bound.

In order to prove (ii), we first exhibit an assignment
$N$ for $Q_k$ that achieves interference
$k+2$. We construct $N$ as follows: first, for
$i = 2, \dots, k+1$, we take for $R_i$
the assignment $N_i$ from
Lemma~\ref{lem:opt_ass} whose root is
the closest point of $P_i$ to $a$. Then, we
connect $a$ to the closest point in $R_2$, and
for $i = 2, \dots, k+1$,
we connect the root of $R_i$ to the root of
$R_{i+1}$. Using the properties
from Lemma~\ref{lem:opt_ass}, we can check
that this
assignment has interference $k+2$.

Next, we show that all valid assignments for
$Q_k$ have interference at least $k+2$.
Let $N$ be an assignment for $Q_k$.
Let $p$ be the leftmost point of
$R_{k+2}$, and let $q$ be the last point on the
path from $p$ to the root of $N$ that lies
in $R_{k+2}$.
We change the assignment $N$ such that all
edges leaving $R_{k+2}$ now go to
$q$. This yields a valid assignment $\widetilde{N}$
for $R_{k+2}$ with root $q$. Thus, $I(\widetilde{N}) \geq k+2$,
by Proposition~\ref{prop:lowerb}. Hence,
by construction, $I(N) \geq I(\widetilde{N}) \geq k+2$, since
$d_k \geq \ell_{k+2}$.

For (iii), let $N$ be an optimal assignment for $Q_k$.
We prove by induction that the root of
$N$ lies in $R_{k+2}$, and that $N$ has $k$ bends,
all of which originate outside of $R_{k+2}$.
As argued above, we have $I(N) = k+2$.
As before, let $p$ be the leftmost point of $R_{k+2}$
and $q$ the last point on the path from $p$ to
the root of $G_{Q_k}(N)$. Suppose that $q$ is not
the root of $N$. Then $q$ has an outgoing edge that
increases the interference of all points in $R_{k+2}$
by $1$. Furthermore, by constructing a valid assignment
$\widetilde{N}$ for $R_{k+2}$ as in the previous paragraph,
we see that the interference in $N$ of all edges that originate
from $P_{k+2} \setminus q$
is at least $k+2$. If follows that $I(N) \geq k + 3$,
although $N$ is optimal.

Thus, the root $r$ of $N$ lies in $R_{k+2}$. Let $b$ be a
point outside $R_{k+2}$ with $N(b) \in R_{k+2}$.
The outgoing edge from $b$ increases
the interference of all points in $Q_k \setminus R_{k+2}$
by $1$. Furthermore, we can construct a valid assignment
$\widehat{N}$ for $Q_k \setminus R_{k+2}$ by redirecting all
edges leaving $Q_{k-1}$ to $b$. By construction, $I(\widehat{N}) \leq k+1$,
so by (ii), $\widehat{N}$ is optimal for $Q_{k-1}$
with interference $k+1$. By induction, $\widehat{N}$ has its root
in $R_{k+1}$ and has at least $k-1$ bends, all of which originate
outside $R_{k+1}$. Thus, $b$ must lie in $R_{k+1}$.
Since $b$ was arbitrary, it follows that all bends of $\widehat{N}$
are also bends of $N$. The edge from $b$ in $N$ is
also a bend, so the claim follows.
\qed{}
\end{proof}

\section{Conclusion}

We have shown that interference minimization in two-dimensional planar sensor
networks is NP-complete. In one dimension, there exists an algorithm
that runs in quasi-polynomial time, based on the fact that there are always
optimal solutions with the BST-property. Since it is generally believed that
NP-complete problems do not have quasi-polynomial algorithms, our result indicates
that one-dimensional interference minimization is probably not NP-complete. However,
no polynomial-time algorithm for the problem is known so far. Furthermore,
our structural result in Section~\ref{sec:further} indicates that optimal solutions
can exhibit quite complicated behavior, so further ideas will be necessary for a
better algorithm.

\section*{Acknowledgments}
We would like to thank Maike Buchin, Tobias Christ, Martin Jaggi, Matias Korman, Marek Sulovsk\'y,
and Kevin Verbeek for fruitful discussions.

\bibliographystyle{abbrv}
\bibliography{interfere}

\newpage
\appendix
\section{Omitted Proofs}\label{sec:omitted}

\begin{proof}[of Lemma~\ref{lem:outedge}]
By Lemma~\ref{lem:gadget_NN}, both
$M$ and $I_c$ in $P_v$ have interference
at least $5$.
This implies that neither $M$, nor $C$, nor
any point in the inhibitor of $P_v$ can be
incident to an outgoing edge of $P_v$:
such an edge would increase the interference at
$M$ or at $I_c$.
In particular, note that the
distance between the inhibitors in
two distinct vertex gadgets is at least
$\sqrt{2}/2 - O(\eps) > 1/2 +O(\eps)$,
the distance between $M$
and its corresponding inhibitor; see the
dotted line in Fig~\ref{fig:gadget_example}.

Thus, all outgoing edges for $P_v$ must originate in
a satellite station.
If there were a satellite station in $P_v$ where
both points are incident to outgoing edges,
the interference at $M$ would increase.
Furthermore, if there were a satellite station in $P_v$ with
an outgoing edge that does not go the partner station,
this would increase the interference at the main point
of the partner vertex gadget, or at the inhibitor center $I_v$ of $P_v$.
\qed{}
\end{proof}

\begin{proof}[of Proposition~\ref{prop:descendants}]
Since $p_l$ and $p_r$ each have a
path to $p$, the  interval $[p_l, p_r]$ is
covered by edges that begin in proper descendants of $p$.
Thus, if $[p_l, p_r]$ contains a point $z$ that is not a descendant
of $p$, then $z$ would be covered by an edge $p_1 p_2$ with
$p_1$ a proper descendant of $p$.
Thus, $p_1 p_2$ would be a cross edge with larger hop-distance to the
root, despite the choice of $pq$.
\qed{}
\end{proof}

\begin{proof}[of Proposition~\ref{clm:validass}]
We must show that all points in $G_P(\widehat{N})$ can reach the root.
At most two edges change: $pq$ and (potentially)
$z' z$.
First, consider the change of $pq$ to $pz$.
This affects only the descendants of $p$.
Since $z$ is not a descendant of $p$,
the path from $z$ to the root does not use the edge $pq$, and hence
all descendants of $p$ can still reach the root.
Second, consider the change of $z' z$ to an edge from $z'$ to
$p_l$ or $p_r$. Both $p_l$ and $p_r$ have $z$ as ancestor (since we
introduced the edge $pz$), so all descendants of $z'$
can still reach the root.
\qed{}
\end{proof}

\begin{proof}[of Proposition~\ref{prop:lowerb}]
The proof is by induction on $i$. For $P_0$ and $P_1$, the claim
is clear.

Now consider a valid assignment $N$ for $P_i$ with sink $r$. Let $Q$ and
$R$ be the two $P_{i-1}$ subsets of $P_i$, and suppose
without loss of generality that $r \in R$.
Let $E$ be the edges that cross from $Q$ to $R$.
Fix a point $p \in Q$, and let $q$ be the last vertex on the
path from $p$ to $r$ that lies in $Q$. We replace every
edge $ab \in E$ with $a \neq q$
by the edge $aq$. By the definition of $P_i$, this does not increase the interference.
We thus obtain a valid assignment
$N' : Q \rightarrow Q$ with sink $q$ such that
$I(N) \geq I(N') + 1$, since the ball for the edge
between $q$ and $R$ covers all of $Q$. By induction,
we have $I(N') \geq i-1$, so $I(N) \geq i$, as claimed.
\qed{}
\end{proof}

\section{Nearest Neighbor Algorithm and Lower Bound}\label{sec:nna}

First, we prove  that we can always obtain
interference $O(\log n)$, a fact used in Section~\ref{sec:qalgo}.
This is achieved by the \emph{Nearest-Neighbor-Algorithm}
(NNA)~\cite{RickenbachWaZo09,Korman12}.
It works as follows.

At each step, we maintain a partition
$\mathcal{S} = \{S_1, S_2, \ldots, S_k\}$ of $P$, such that
the convex hulls of the $S_i$ are disjoint. Each set $S_i$
has a designated sink $r_i \in S_i$ and an assignment
$N: S_i \rightarrow S_i$ such that the graph
$G_{S_i}(N_i)$ is acyclic and has $r_i$ as the only sink.
Initially, $\mathcal{S}$ consists of $n$ singletons, one
for each point in $P$. Each point in $P$ is the sink of its
set, and the assignments are trivial.

Now we describe how to go from a partition
$\mathcal{S} = \{S_1, \dots, S_k\}$ to a new partition $\mathcal{S}'$.
For each sink $r_i \in S_i$, we define the successor $Q(r_i)$ as the closest
point to $r_i$ in $P \setminus S_i$. We will ensure that this closest
point is unique in every round after the first.
In the first round, we break
ties arbitrarily
Consider the directed graph $R$ that has vertex set $P$ and
contains all edges from the component graphs $G_{S_i}(N_i)$ together
with edges $r_i Q(r_i)$, for $i = 1, \dots, k$. Let $S_1', S_2', \dots, S_{k'}'$ be the
components of $R$. Each such component $S_j'$ contains exactly one cycle,
and each such cycle contains exactly two sinks $r_a$ and $r_{a+1}$.
Pick $r_j' \in \{r_a, r_{a+1}\}$ such that the distances between
$r_j'$ and the closest points in the neighboring components
$S_{j-1}'$ and $S_{j+1}'$ are distinct (if they exist). At least
one of $r_a$ and $r_{a+1}$ has this property, because $r_a$ and
$r_{a+1}$ are distinct. Suppose that $r_j' = r_a$ (the other case
is analogous). We make $r_a$ the new sink of $S_j'$, and we
let $N_j'$ be the union of $r_{a+1} Q(r_{a+1})$ and
the assignments $N_i$ for all
components $S_i \subseteq S_j$. Clearly, $N_j'$ is a valid assignment
for $S_j'$. We set $\mathcal{S}' = \{S_1', \dots, S_{k'}'\}$.
This process continues until a single component remains.

\begin{obs}\label{obs:nna_upper}
The nearest neighbor algorithm ensures interference at most
$\lceil \log n \rceil + 2$.
\end{obs}

\begin{proof}
Since each component in $\mathcal{S}$ is combined with at least one
other component of $\mathcal{S}$, we have $k' \leq \lfloor k/2 \rfloor$,
so there are at most $\lceil \log n \rceil$ rounds.

Now fix a point $p \in P$. We claim that in the interference
of $p$ increases by at most $1$ in each round, except for
possibly two rounds in which the interference increases by $2$.
Indeed, in the first round, the interference increases by at most $2$,
since each point connects to its nearest neighbor (the increase by $2$ can
happen if there is a point with two nearest neighbors).
In the following rounds, if $p$ lies in the interior of a connected component $S_i$,
its interference increases by at most $1$ (through the edge from
$r_i$ to $Q(r_i)$). If $p$ lies on the boundary of
$S_i$, its interference may increase by $2$ (through the edge between
$r_i$ and $Q(r_i)$ and the edge that connects a neighboring component
to $p$). In this case, however, $p$ does not appear on the boundary of
any future components, so the increase by $2$ can happen at most once.
\qed{}
\end{proof}

Next, we show that interference
$\Omega(\log n)$ is sometimes necessary.
We make use of the points sets $P_i$ constructed in Section~\ref{sec:further}.

\begin{corol}
For every $n$, there exists a point set $Q_n$ with $n$ points such that
every valid assignment for $N$ has interference $\lfloor \log n \rfloor$.
\end{corol}

\begin{proof}
Take the point set $P_{\lfloor \log n \rfloor}$ from Section~\ref{sec:further} and add
$n - 2^{\lfloor \log n \rfloor}$ points sufficiently far away.
The bound on the interference follows from Proposition~\ref{prop:lowerb}.
\qed{}
\end{proof}

\end{document}